\documentclass[letterpaper, 10 pt, conference]{ieeeconf}  
\usepackage{undertilde}
\usepackage{graphicx}
\usepackage{graphics} 
\usepackage{epsfig} 
\usepackage{epstopdf}
\usepackage[abs]{overpic}
\usepackage{algorithmic}
\usepackage{setspace}
\usepackage{mathtools}
\usepackage{amsmath,amsfonts,amssymb}

\usepackage{accents}
\usepackage{tikz}
\usepackage{epsfig} 
\usepackage{theorem}
\newtheorem{theorem}{\bf{Theorem}}[section]

\newtheorem{lem}[theorem]{Lemma}

\theoremstyle{plain}

\newcounter{algno} 
\newenvironment{definition}[1][Definition]{\begin{trivlist}
\item[\hskip \labelsep {\bfseries #1}]}{\end{trivlist}}
\newenvironment{remark}[1][Remark]{\begin{trivlist}
\item[\hskip \labelsep {\bfseries #1}]}{\end{trivlist}}

\newcounter{asno} 
\title{\LARGE \bf Resilience of Locally Routed Network Flows: More Capacity is Not Always Better
}

\author{A. Yasin Yaz{\i}c{\i}o\u{g}lu, Mardavij Roozbehani, and Munther A. Dahleh \\
Laboratory for Information and Decision Systems\\
Massachusetts Institute of Technology, Cambridge, MA 02139, USA\\
{\tt\small yasiny@mit.edu, mardavij@mit.edu, dahleh@mit.edu}
}

\IEEEoverridecommandlockouts
\begin{document}
\maketitle
\begin{abstract}
%



In this paper, we are concerned with the resilience of locally routed network flows with finite link capacities. In this setting, an external inflow is injected to the so-called origin nodes. The total inflow arriving at each node is routed locally such that none of the outgoing links are overloaded unless the node receives an inflow greater than its total outgoing capacity. A link irreversibly fails if it is overloaded or if there is no operational link in its immediate downstream to carry its flow. For such systems, resilience is defined as the minimum amount of reduction in the link capacities that would result in the failure of all the outgoing links of an origin node. We show that such networks do not necessarily become more resilient as additional capacity is built in the network. Moreover, when the external inflow does not exceed the network capacity, selective reductions of capacity at certain links can actually help averting the cascading failures, without requiring any change in the local routing policies. This is an attractive feature as it is often easier in practice to reduce the available capacity of some critical links than to add physical capacity or to alter routing policies, e.g., when such policies are determined by social behavior, as in the case of road traffic networks. The results can thus be used for real-time monitoring of distance-to-failure in such networks and devising a feasible course of actions to avert systemic failures.

\end{abstract}


\section{Introduction}
\label{intro}
Resilience is a critical aspect in the design and operation of infrastructure systems
such as transportation, power, water, and communication networks. In many applications, the network faces various disturbances and operates under the actions taken by agents with limited information about the system. Such an operation of the network often results in a suboptimal global performance (e.g., \cite{Roughgarden02}), and it may even lead to cascading failures with severe systemic consequences (e.g., \cite{Como13b,Savla14}). Accordingly, there has been a significant research interest in modeling cascading failures on networks (e.g., \cite{Granovetter78,Adler91,Liggett12, Watts02, Dobson05}), investigating the influence of network topology on failure propagation (e.g., \cite{Blume11,Acemoglu12}), and designing networks (e.g., \cite{Ash07,Yasin15TNSE}) and control policies (e.g, \cite{Como13b,Savla14,Motter04,Varaiya11}) for resilient operation. 



 In this paper, we focus on locally routed network flows with finite link capacities. We consider a dynamical model that is similar to the ones presented in \cite{Como13b,Savla14}, where the authors investigate the resilience of network flows under some local routing policies. In this setting, a constant external inflow is injected to the network, and the inflow at each node is routed locally such that none of the outgoing links are overloaded unless the node receives an inflow greater than its total outgoing capacity.
A link irreversibly fails (i.e., its capacity drops to zero) if it is overloaded or if there is no operational link in its immediate downstream to carry its flow. Accordingly, a cascade of failures can be induced by a sufficiently large disturbance that reduces the link capacities. However, we show in this paper that the link capacities do not have a monotonic influence on the \textcolor{black}{resilience} of locally routed network flows. Furthermore, cascading failures can be attributed to myopic routing decisions when the external inflow does not exceed the network capacity, and they can actually be avoided by reducing the capacities of some critical links without requiring any alteration of the local routing policies. Such a counterintiutive influence of the link capacities stems from the lack of global information rather than the lack of cooperative decisions captured by the Braess' paradox \cite{Braess05}. The results of this paper pave the way for designing dynamic capacity allocation policies to avoid the systemic failures due to local routing in network flows.

The organization of this paper is as follows: Section \ref{prelim} provides some preliminaries. Section \ref{model} presents the proposed model of locally routed network flows. Section \ref{main} presents our main results. Finally, Section \ref{conclusion} concludes the paper.

\section{Preliminaries}
\label{prelim}

\subsection{Notation}
For any finite set $A$ with cardinality $|A|$, we use $\mathbb{R}^A$ ($\mathbb{R}_{+}^A$, $\mathbb{R}_{++}^A$) to denote 
the space of real-valued (nonnegative-real-valued, positive-real-valued) ${|A|-\mbox{dimensional}}$ vectors whose components are indexed by the elements of $A$. Accordingly, for any $a\in A$ and $x\in\mathbb{R}^A$, $x_a \in\mathbb{R}$ denotes the corresponding entry of $x$, and we define $x_{a'} = 0$ for all $a' \notin A$. Similarly, for any $A' \subseteq A$, we use $x_{A'} \in\mathbb{R}^{A'}$ to denote the $|A'|$-dimensional vector consisting only of the components of $x$ whose indices are in $A'$. \textcolor{black}{For any pair of vectors $\underaccent{\bar}{x},\bar{x} \in \mathbb{R}^A$, we use $[\underaccent{\bar}{x},\bar{x}]$ to denote the set of vectors $x\in \mathbb{R}^A$ such that $\underaccent{\bar}{x} \leq x \leq \bar{x}$. } The all-one and all-zero vectors, their size being clear from the context, will be denoted by $\bold{1}$ and $\bold{0}$, respectively. 


\subsection{\textcolor{black}{Graph basics}}
A directed graph, $\mathcal{G}=(\mathcal{V},\mathcal{E})$, consists of a set of nodes, $\mathcal{V}$, and a set of edges, $\mathcal{E} \subseteq \mathcal{V} \times \mathcal{V}$, given by ordered pairs of nodes. A graph is a multi-graph if multiple edges are allowed between the nodes, i.e., if $\mathcal{E}$ is a multi-set. Each ${(v,w)\in\mathcal{E}}$ denotes a link from $v$ (the tail) to $w$ (the head). For each $v \in \mathcal{V}$, we use $\mathcal{E}_v^-$ and $\mathcal{E}_v^+$ to denote the corresponding sets of incoming and outgoing links, respectively. Similarly, for any $\mathcal{U} \subseteq \mathcal{V}$, 
\begin{equation}
\label{uminus}
\mathcal{E}_\mathcal{U}^-= \{(v,w)\in \mathcal{E} \mid v \notin\mathcal{U},  w \in\mathcal{U} \},
\end{equation}
\begin{equation}
\label{uplus}
\mathcal{E}_\mathcal{U}^+= \{(v,w)\in \mathcal{E} \mid v \in\mathcal{U},  w \notin\mathcal{U} \}.
\end{equation}
Also, for  any $v \in \mathcal{V}$, we use $\mathcal{V}_v^-$ and $\mathcal{V}_v^+$ to denote the tails of incoming links and the heads of the outgoing links, respectively.

A path is a sequence of nodes such that for any two nodes $v,w \in \mathcal{V}$ that are consecutive in  the sequence, $w \in \mathcal{V}_v^+$. A directed graph is acyclic if, for every node $v$, there exists no feasible path that starts and ends at $v$.

\begin{remark}:
\label{monlab} Any acyclic \textcolor{black}{ graph} of $n$ nodes can be represented as $\mathcal{G}=(\mathcal{V},\mathcal{E})$ such that $V= \{1, 2, \hdots, n\}$ and 
\begin{equation}
\label{monlab1}
\mathcal{E}_v^- \subseteq \bigcup_{1\leq u < v}\mathcal{E}_u^+, \; \forall v \in \mathcal{V}.
\end{equation} 
In the remainder of this paper, such a labeling of the nodes will be assumed whenever an acyclic  \textcolor{black}{ graph} is considered.  
\end{remark}

\section{System Model}
\label{model}
In this section, we provide some definitions and the dynamical model of locally routed network flows with capacity constraints.

\begin{definition} (\emph{Flow Network}):  A flow network is a directed multi-graph $\mathcal{G}=(\mathcal{V},\mathcal{E})$ such that there exists a path from any node to some destination node $v \in \mathcal{V}_D$, where 
\begin{equation}
\label{vd}
\mathcal{V}_D = \{v\in \mathcal{V} \mid \mathcal{E}_v^+ = \emptyset \}.
\end{equation}
Furthermore, the set of origin nodes $\mathcal{V}_O$ is defined as
\begin{equation}
\label{vo}
\mathcal{V}_O = \{v\in \mathcal{V} \mid \mathcal{E}_v^- = \emptyset \}.
\end{equation}
The nodes that are neither an origin nor a destination constitute the set of intermediate nodes  $\mathcal{V}_I =\mathcal{V} \setminus \mathcal{V}_O \setminus \mathcal{V}_D$.
\end{definition}

\begin{definition} (\emph{Local Routing Policy}):  
Given a flow network $\mathcal{G}=(\mathcal{V}, \mathcal{E})$, a local routing policy $\mathcal{R}$ is a family of functions 
\begin{equation}
\label{lrp1}
\mathcal{R}^v: \mathbb{R}_{+}^{\mathcal{E}_v^+} \times \mathbb{R}_{+} \mapsto \mathbb{R}_{+}^{\mathcal{E}_v^+}, \; v \in \mathcal{V} \setminus \mathcal{V}_D, 
\end{equation}
such that, for any vector of link capacities $C \in \mathbb{R}_{+}^{\mathcal{E}} $ and vector of inflows $\mu  \in \mathbb{R}_{+}^\mathcal{V} $,

\begin{equation}
\label{lrp2}\sum_{e\in \mathcal{E}_v^+}\mathcal{R}^v_e(C_{\mathcal{E}_v^+},\mu_v)= \mu_v, 
\end{equation}
\begin{equation}
\label{lrp4}
\textcolor{black}{\mu_v \leq \bold{1}^\mathsf{T} C_{\mathcal{E}_v^+} \Rightarrow
\mathcal{R}^v(C_{\mathcal{E}_v^+}, \mu_v) \leq C_{\mathcal{E}_v^+}.   }
\end{equation}
\end{definition}

 Local routing policies as defined above map the local inflows and link capacities to the corresponding equilibrium outflows. Accordingly, any local routing policy $\mathcal{R}$ can be considered as a family of functions such that each $\mathcal{R}^v(C_{\mathcal{E}_v^+},\mu_v)$ is an optimizer of some local objectives when the conservation of flow and the local capacity constraints admit feasible outflows.

\emph{Example (Proportional routing)}: One example of a local routing policy is the proportional routing, that is, for any ${v \in \mathcal{V} \setminus \mathcal{V}_D}$ and $e \in \mathcal{E}_v^+$
\begin{equation}
 \label{propor}
\mathcal{R}^v_e(C_{\mathcal{E}_v^+},\mu_v) =\dfrac{\mu_v C_e}{\bold{1}^\mathsf{T}  C_{\mathcal{E}_v^+} }, \forall C_{\mathcal{E}_v^+} \neq \bold{0}. 
 \end{equation}
For instance, in a transportation network, if each driver myopically prefers the least congested (smallest ratio of flow to capacity) link in the immediate downstream, then any feasible inflow (i.e., ${\mu_v \leq \bold{1}^\mathsf{T} C_{\mathcal{E}_v^+}}$) to a non-destination node results in a unique equilibrium outflow as given by the proportional routing.

\begin{definition} (\emph{Locally Routed Network Flow}): A locally routed network flow  consists of a flow network $\mathcal{G}= (\mathcal{V},\mathcal{E})$, a vector of link capacities $C \in \mathbb{R}_{+}^\mathcal{E}$, a local routing policy $\mathcal{R}$, and a constant external inflow $\lambda \in \mathbb{R}_{+}^{\mathcal{V}_O}$. Accordingly, we use $(\mathcal{G}, C, \mathcal{R}, \lambda)$ to denote a dynamical system with the state $(C(t), f(t))$ , where $C(t) \in \mathbb{R}_{+} ^ \mathcal{E}$ denotes the available link capacities and $f(t) \in \mathbb{R}_{+} ^ \mathcal{E}$ is the flow vector such that  

\begin{equation}
\label{dyn0}
C(0) = C, \; f(0)= \bold{0},
\end{equation} 
\begin{equation}
\label{dyn1}
\mu_w(t)= \left\{\begin{array}{ll}\lambda_w, & \mbox{ if  } w \in \mathcal{V}_O, \\ \sum_{e \in \mathcal{E}_{w}^-}f_e(t), & \mbox{ o.w. }\end{array}\right.,
\end{equation} 
and, for every $e=(v,w) \in \mathcal{E}$,
\begin{equation}
\label{dyn2}
\big(C_e(t+1), f_e(t+1)\big)=\left\{\begin{array}{l}(0,0), \mbox{ if  (\ref{fail1}) or (\ref{fail2}) holds}, \\ \big(C_e(t), \mathcal{R}^v_e(C_{\mathcal{E}_v^+}(t),\mu_v(t)) \big), \mbox{ o.w. }\end{array}\right.
\end{equation}
where the link-failure conditions are
\begin{equation}
\label{fail1}
 \mathcal{R}^v_e(C_{\mathcal{E}_v^+}(t),\mu_v(t)) > C_e(t),
\end{equation}
\begin{equation}
\label{fail2}
f_e(t)>0, \quad \bold{1}^\mathsf{T} C_{\mathcal{E}_w^+}(t) =0.
\end{equation}

\end{definition}

For any flow network $\mathcal{G}=(\mathcal{V},\mathcal{E})$, ${C \in \mathbb{R}_{+}^\mathcal{E}}$, and $\lambda \in \mathbb{R}_{+}^{\mathcal{V}_O}$,  the set of feasible balanced flows, $\mathcal{F}_\mathcal{G} (C,\lambda )$, consists of feasible flow vectors for which the inflow is equal to the outflow for every non-destination node, i.e., 
\begin{multline}
\label{fb}
\mathcal{F}_\mathcal{G} (C,\lambda ) = \{ f \in \mathbb{R}_{+} ^\mathcal{E} \mid f \leq C,  \\ \;  \lambda_v + \sum_{e \in \mathcal{E}_{v}^-}f_e = \sum_{e \in \mathcal{E}_v^+}f_e, \forall v \in \mathcal{V} \setminus \mathcal{V}_D\}.
\end{multline} 
Accordingly, 
\begin{equation}
\label{fbu}
\sum _{v \in \mathcal{U}}\lambda_v+ \sum_{e  \in \mathcal{E}_\mathcal{U} ^-}f_e = \sum_{e  \in \mathcal{E}_\mathcal{U} ^+}f_e,  \; \forall f \in \mathcal{F}_\mathcal{G} (C,\lambda ),  \forall \mathcal{U} \subseteq \mathcal{V} \setminus \mathcal{V}_D.
\end{equation} 




\begin{definition} (\emph{Transferring Network Flow}): 
A locally routed network flow $(\mathcal{G}, C, \mathcal{R}, \lambda)$ is transferring if ${\lim_{t \to \infty} (C(t),f(t))=(C^*,f^*)}$ such that 
\begin{equation}
\label{tr}
f^* \in \mathcal{F}_\mathcal{G} (C^*,\lambda).
\end{equation} 
 Furthermore, such a system is called safely transferring if none of the links fail, i.e.,
\begin{equation}
\label{tr2}
C^* =C.
\end{equation}

%

\end{definition}

\begin{definition} (\emph{Resilience}): 
The resilience of any locally routed network flow $(\mathcal{G}, C, \mathcal{R}, \lambda)$ is defined as the minimum amount of capacity reduction that would induce a non-transferring system, i.e.,
\begin{equation}
\label{res}
\min_{\delta \in \Delta(\mathcal{G},C,\mathcal{R},\lambda)} \bold{1}^\mathsf{T}\delta,
\end{equation} 
where $\Delta(\mathcal{G},C,\mathcal{R},\lambda)$ is
\begin{equation}
\label{dis}
\Delta(\bullet)= \{ \bold{0} \leq \delta \leq C \mid \mbox{($\mathcal{G},C-\delta,\mathcal{R},\lambda)$ is non-transferring}\}.
\end{equation}
As such, resilience quantifies the distance to systemic failures, and it is equal to zero for non-transferring networks.
 \end{definition} 

\section{Main Results}
\label{main}
In this section, we present our main results related to the influence of link capacities on locally routed network flows. In particular, we focus on acyclic networks and show the following:

\begin{enumerate}
\item The link capacities do not have a monotonic influence on the resilience of locally routed networks. For a broad family of such systems, a non-transferring network can be induced by increasing the capacities.
\item When the external inflow does not exceed the network capacity, cascading failures can actually be avoided via a selective reduction of the link capacities, without requiring any alteration of the local routing policies.
\end{enumerate}

We start our analysis by presenting a sufficient condition on the link capacities for ensuring a safely transferring system under any local routing policy. This sufficient condition will later be used in our main results, Theorems \ref{capincbad} and \ref{redt}.
\begin{lem}
\label{transt}
For any acyclic flow network  $\mathcal{G}= (\mathcal{V},\mathcal{E})$, link capacities $C \in \mathbb{R}_{+}^\mathcal{E}$, and external inflow $\lambda \in \mathbb{R}_{+}^{\mathcal{V}_O}$, if
\begin{equation}
\label{transt2}
\lambda_v \leq \bold{1}^\mathsf{T}  C_{\mathcal{E}_v^+}, \forall v \in  \mathcal{V}_O,
\end{equation} 
\begin{equation}
\label{transt1}
\bold{1}^\mathsf{T}  C_{\mathcal{E}_v^-} \leq \bold{1}^\mathsf{T}  C_{\mathcal{E}_v^+}, \forall v \in  \mathcal{V}_I,
\end{equation} 
then $(\mathcal{G}, C, \mathcal{R}, \lambda)$ is safely transferring for any local routing policy $\mathcal{R}$.
\end{lem}

\begin{proof}
 First, we show that, for any local routing policy $\mathcal{R}$, (\ref{transt2}) and (\ref{transt1}) together imply $C(t)=C$ for all $t\geq0$ under (\ref{dyn0})-(\ref{dyn2}). For the sake of contradiction, let $e=(v,w)$ be one of the first links to fail, i.e., 
 \begin{equation}
\label{transt3a}
C_e>0, \;C_e(\tau)=0,
\end{equation} 
 where
\begin{equation}
\label{transt3}
\tau =  \min(\{t \geq 0 \mid C(t) \neq C\}).
\end{equation} 
Then, either (\ref{fail1}) or (\ref{fail2}) holds for $e$ at ${t=\tau-1}$. 

1) Suppose that (\ref{fail1}) holds. Then, (\ref{lrp2}) and (\ref{lrp4}) imply that ${\mu_v(\tau-1)> \bold{1}^\mathsf{T}  C_{\mathcal{E}_v^+}}(\tau-1)$, and, due to  (\ref{transt3}), ${C_{\mathcal{E}_v^+}(\tau-1)= C_{\mathcal{E}_v^+}}$. Hence, in light of (\ref{dyn1}) and (\ref{transt2}), ${v\notin \mathcal{V}_O}$. As such, $v\in \mathcal{V}_I$ and, due to (\ref{transt1}), ${\mu_v(\tau-1)> \bold{1}^\mathsf{T}  C_{\mathcal{E}_v^+}}$  implies
\begin{equation}
\label{transt33}
\mu_v(\tau-1)=\sum_{e\in \mathcal{E}_v^- }f_e(\tau-1) > \bold{1}^\mathsf{T}  C_{\mathcal{E}_v^-}.
\end{equation} 
Note that  (\ref{transt33}) contradicts with (\ref{dyn2}) and (\ref{fail1}), which together imply 
\begin{equation}
\label{transt333}
f_e(t) \leq C_e(t) \leq C_e, \forall e \in \mathcal{E}, \forall t \geq 0. 
\end{equation} 

2) Suppose that (\ref{fail2}) holds. Then, either $\bold{1}^\mathsf{T} C_{\mathcal{E}_w^+} = 0$ or there exists some $\tau'<\tau$ such that $\bold{1}^\mathsf{T} C_{\mathcal{E}_w^+}(\tau'-1)>0$ and $\bold{1}^\mathsf{T} C_{\mathcal{E}_w^+}(\tau')=0$. Note that the latter contradicts with (\ref{transt3}). Furthermore, if $\bold{1}^\mathsf{T} C_{\mathcal{E}_w^+} = 0$, then (\ref{transt1}) implies that $C_e=0$, which contradicts with (\ref{transt3a}). 

Consequently, both (\ref{fail1}) and (\ref{fail2}) lead to contradiction, and $C(t) = C$ for all $t\geq 0$. In that case, since $\mathcal{G}$ is acyclic, (\ref{lrp2}) and
(\ref{dyn0})-(\ref{dyn2}) together imply 
\begin{equation}
\label{transt9}
f(t+1)=  f(t)  \in  \mathcal{F}_\mathcal{G} (C,\lambda ), \; \forall t\geq l_{max}(\mathcal{V}_O,\mathcal{V}_D),
\end{equation} 
where $l_{max}(\mathcal{V}_O,\mathcal{V}_D)$ denotes the length of the longest feasible path on $\mathcal{G}$ from any origin node to any destination node. Consequently, for any acyclic $\mathcal{G}$, if $C$ and $\lambda$ satisfy (\ref{transt2}) and (\ref{transt1}), then $(\mathcal{G}, C, \mathcal{R}, \lambda)$ is safely transferring for any local routing policy $\mathcal{R}$.

\end{proof}

\subsection{Cascading Failures Induced by More Capacity}
\textcolor{black}{
Our next result shows that increasing the link capacities can lead to a non-transferring system for a broad family of locally routed network flows. Specifically, having more capacity can cause a systemic failure if the network has a node ${v \in \mathcal{V}}$ such that 1) $v$ has outgoing links to multiple nodes, at least one of which is an intermediate node, and 2) increasing the capacities of any links in $\mathcal{E}_v^+$ would reduce the amount of flow routed to the other operational links in $\mathcal{E}_v^+$ whose capacities remain the same. We say that a local routing policy is capacity-monotone at $v$ if it satisfies the second property. For instance, the proportional routing (see (\ref{propor})) is capacity-monotone at every non-destination node. }


\begin{definition} (\emph{Capacity-Monotone Local Routing}):  
 A local routing policy $\mathcal{R}$ is capacity-monotone at node $v$ if 
\begin{equation}
\label{cib1}
 \mathcal{R}^v_e(\tilde{C}_{\mathcal{E}_v^+},\mu_v) <\mathcal{R}^v_{e}(C_{\mathcal{E}_v^+},\mu_v), \;  \forall  e \in \mathcal{E}_v^+:\tilde{C}_e =C_e >0,
\end{equation}
for every ${\mu_v  \in \mathbb{R}_{++}}$ and ${\tilde{C}_{\mathcal{E}_v^+}\geq C_{\mathcal{E}_v^+}  \in \mathbb{R}_{+}^{\mathcal{E}_v^+}}$ such that $\tilde{C}_{\mathcal{E}_v^+} \neq C_{\mathcal{E}_v^+}$.
\end{definition}

\begin{theorem}
\label{capincbad}
Let $\mathcal{G}= (\mathcal{V},\mathcal{E})$ be an acyclic flow network, and let $\mathcal{R}$ be a local routing policy. There exist link capacities ${\tilde{C}  \geq C \in \mathbb{R}_{+}^\mathcal{E}}$ such that $(\mathcal{G}, C, \mathcal{R}, \lambda)$ is transferring but $(\mathcal{G}, \tilde{C}, \mathcal{R}, \lambda)$ is non-transferring for some $\lambda \in \mathbb{R}_{+}^{\mathcal{V}_O}$, if there exists $v\in \mathcal{V}$ such that \textcolor{black}{it has outgoing links to multiple nodes, at least one of which is an intermediate node (i.e., $|\mathcal{V}_v^+| \geq 2$, $ \mathcal{V}_v^+ \not \subseteq \mathcal{V}_D$), and $\mathcal{R}$ is capacity-monotone at $v$. } 
\end{theorem}
\begin{proof}

Let $\mathcal{G}= (\mathcal{V},\mathcal{E})$ be an acyclic flow network, and let $v\in \mathcal{V}$ be a node such that $|\mathcal{V}_v^+| \geq 2$, $ \mathcal{V}_v^+ \not \subseteq \mathcal{V}_D$ and the local routing policy $\mathcal{R}$ is capacity-monotone at $v$. Consider any $\lambda \in \mathbb{R}_{++}^{\mathcal{V}_O}$ and $C \in \mathbb{R}_{++}^\mathcal{E}$ that satisfies 
\begin{equation}
\label{cib2}
\bold{1}^\mathsf{T}  C_{\mathcal{E}_u^+} = \left\{\begin{array}{ll} \lambda_u, &\mbox{ if } u\in\mathcal{V}_O, \\ \bold{1}^\mathsf{T}  C_{\mathcal{E}_u^-}, &\mbox{ if } u\in\mathcal{V}_I. \end{array}\right. 
\end{equation}
Since $\mathcal{G}$ is acyclic, an example of such $C \in \mathbb{R}_{++}^\mathcal{E}$ can be constructed iteratively by assigning some $C_{\mathcal{E}_u^+}> \bold{0}$ satisfying (\ref{cib2}) for ${u=1, 2, \hdots, \max(\mathcal{V} \setminus \mathcal{V}_D)}$. Let ${w = \min(\mathcal{V}_v^+)}$. Since the node labels satisfy (\ref{monlab1}) and ${\mathcal{V}_v^+ \not \subseteq \mathcal{V}_D}$, we have $w \in \mathcal{V}_I$. Consider any $\tilde {C}\geq C$ such that
\begin{equation}
\label{cib3}
\tilde{C}_e \left\{\begin{array}{ll}  = C_e, &\mbox{ if } e \notin  \mathcal{E}_w^- \cap \mathcal{E}_v^+, \\  > C_e, & \mbox{ otherwise. } \end{array}\right.
\end{equation}
In the remainder, we use $f(t),C(t),\mu(t)$ and $\tilde{f}(t),\tilde{C}(t), \tilde{\mu}(t)$ to denote the corresponding trajectories for $(\mathcal{G}, C, \mathcal{R}, \lambda)$ and $(\mathcal{G}, \tilde{C}, \mathcal{R}, \lambda)$, respectively.  We will show that $(\mathcal{G}, C, \mathcal{R}, \lambda)$ is transferring and $(\mathcal{G}, \tilde{C}, \mathcal{R}, \lambda)$ is non-transferring.


In light of Lemma \ref{transt}, (\ref{cib2}) implies that $(\mathcal{G}, C, \mathcal{R}, \lambda)$ is safely transferring. Furthermore,
\begin{equation}
\label{cib3b}
\sum_{e  \in \mathcal{E}_{\{1,2, \hdots, u\}} ^+} C_e=  \sum _{k \in \{1,2, \hdots, u\}}\lambda_k, \forall u \in \mathcal{V} \setminus \mathcal{V}_D.
\end{equation}
Hence, due to (\ref{fbu}) and (\ref{cib3b}), there is only one feasible balanced flow for this configuration, i.e.,  $\mathcal{F}_\mathcal{G} (C,\lambda ) = \{C\}$. As such, since $\mathcal{G}$ is acyclic, (\ref{lrp2}) and
(\ref{dyn0})-(\ref{dyn2}) together imply 
\begin{equation}
\label{cib3c}
(C(t),f(t))= (C,C), \; \forall t\geq l_{max}(\mathcal{V}_O,\mathcal{V}_D),
\end{equation} 
where $l_{max}(\mathcal{V}_O,\mathcal{V}_D)$ denotes the length of the longest feasible path on $\mathcal{G}$ from any origin node to any destination node.


Next, we will show that $(\mathcal{G}, \tilde{C}, \mathcal{R}, \lambda)$ is non-transferring.  To this end, we first show that $(\mathcal{G}, \tilde{C}, \mathcal{R}, \lambda)$ is not safely transferring. For the sake of contradiction, assume that $\tilde{C}(t) = \tilde{C}$ for all $t \geq 0$. In that case, since $\tilde{C}_{\mathcal{E}_u^+}(t) =  C_{\mathcal{E}_u^+}(t)=C_{\mathcal{E}_u^+}$ for all $u \in \mathcal{V}\setminus \{v\}$ and routing only depends \textcolor{black}{on the local inflows and outgoing capacities as in (\ref{lrp1})}, the dynamics in (\ref{dyn0})-(\ref{dyn2}) lead to
\begin{equation}
\label{cib4}
\tilde{\mu}_u(t) = \mu_u(t), \forall u \in  \mathcal{V} \setminus \mathcal{V}_R(v),
\end{equation}
where $\mathcal{V}_R(v)$ denotes the set of nodes that can be reached from $v$ through a non-trivial (traversing at least one edge) directed path on $\mathcal{G}$. Furthermore, since $\mathcal{G}$ is acyclic and ${w = \min(\mathcal{V}_v^+)}$, (\ref{monlab1}) implies
\begin{equation}
\label{cib5}
\mathcal{V}_w^- \cap \mathcal{V}_R(v) = \emptyset.
\end{equation}
Hence, due to (\ref{cib4}) and (\ref{cib5}),
\begin{equation}
\label{cib6}
\tilde{\mu}_u(t) = \mu_u(t), \forall u \in  \mathcal{V}_w^-.
\end{equation}
Due to (\ref{cib2}), (\ref{cib3c}), and (\ref{cib6}), there exists some $\tau \geq 0$ such that
\begin{equation}
\label{cib6b}
\tilde{\mu}_u(\tau) = \bold{1}^\mathsf{T}  C_{\mathcal{E}_u^+}, \forall u \in  \mathcal{V}_w^-.
\end{equation}
Accordingly, 
\begin{equation}
\label{cib7}
\tilde{\mu}_w(\tau+1) =  \sum_{e \in \mathcal{E}_w^- \setminus \mathcal{E}_v^+} C_e+ \sum_{e \in \mathcal{E}_w^- \cap \mathcal{E}_v^+} \mathcal{R}^v_e(\tilde{C}_{\mathcal{E}_v^+},\bold{1}^\mathsf{T}  C_{\mathcal{E}_v^+}).
\end{equation}
Since $C \in \mathbb{R}_{++}^\mathcal{E}$,  (\ref{cib3}) implies $\tilde{C}_e =C_e >0$ for all ${e\in \mathcal{E}_v^+ \setminus \mathcal{E}_w^-}$. Hence, since $\mathcal{R}$ is capacity-monotone at $v$,  
\begin{equation}
\label{cib8}
 \mathcal{R}^v_e(\tilde{C}_{\mathcal{E}_v^+},\bold{1}^\mathsf{T}  C_{\mathcal{E}_v^+})<\mathcal{R}^v_e(C_{\mathcal{E}_v^+},\bold{1}^\mathsf{T}  C_{\mathcal{E}_v^+})  =  C_e, \forall  e \in \mathcal{E}_v^+ \setminus \mathcal{E}_w^-.
\end{equation}
Note that $\mathcal{E}_v^+ \setminus \mathcal{E}_w^- \neq \emptyset$ since $|\mathcal{V}_v^+| \geq 2$. As such, (\ref{lrp2}) and (\ref{cib8}) together imply 
\begin{equation}
\label{cib8b}
 \sum_{e \in \mathcal{E}_w^- \cap \mathcal{E}_v^+} \mathcal{R}^v_e(\tilde{C}_{\mathcal{E}_v^+},\bold{1}^\mathsf{T}  C_{\mathcal{E}_v^+})> \sum_{e \in \mathcal{E}_w^- \cap \mathcal{E}_v^+} C_e.
\end{equation}

Due to (\ref{cib7}) and (\ref{cib8b}), $\tilde{\mu}_w(\tau+1)>  \bold{1}^\mathsf{T}  C_{\mathcal{E}_w^-}$. Furthermore, in light of (\ref{cib2}) and (\ref{cib3}),  ${\bold{1}^\mathsf{T}  C_{\mathcal{E}_w^-}= \bold{1}^\mathsf{T}  C_{\mathcal{E}_w^+}=\bold{1}^\mathsf{T}  \tilde{C}_{\mathcal{E}_w^+} }$. Hence, ${\tilde{\mu}_w(\tau+1)>\bold{1}^\mathsf{T}  \tilde{C}_{\mathcal{E}_w^+}}$ and, due to (\ref{lrp2}), at least one link in $\mathcal{E}_w^+$ fails at $t=\tau+2$. Accordingly, ${\bold{1}^\mathsf{T}  \tilde{C}_{\mathcal{E}_w^+}(t)< \bold{1}^\mathsf{T}  \tilde{C}_{\mathcal{E}_w^+}}$ for all ${t \geq \tau+2}$ under (\ref{dyn2}). As such, $(\mathcal{G}, \tilde{C}, \mathcal{R}, \lambda)$ is not safely transferring. Since $\bold{1}^\mathsf{T}  \tilde{C}_{\mathcal{E}_u^+}= \lambda_u$ for all $u\in \mathcal{V}_O$ and $\bold{1}^\mathsf{T}  \tilde{C}_{\mathcal{E}_u^+}=\bold{1}^\mathsf{T}  \tilde{C}_{\mathcal{E}_u^-}$ for all $u \in \mathcal{V}_I \setminus \{w\}$, the first link to fail indeed has to be in $\mathcal{E}_w^+$. Moreover,  due to (\ref{cib3}) and (\ref{cib3b}), 

\begin{equation}
\label{cib9}
\bold{1}^\mathsf{T}  \tilde{C}_{\mathcal{E}_w^+}(t)< \bold{1}^\mathsf{T}  \tilde{C}_{\mathcal{E}_w^+} \Rightarrow \sum_{e  \in \mathcal{E}_{\{1,2, \hdots, w\}} ^+} \tilde{C}_e(t) <  \sum _{k \in \{1,2, \hdots, w\}}\lambda_k.
\end{equation}

Hence, (\ref{fbu}) and  (\ref{cib9}) together imply that there is no feasible balanced flow on the remaining network once any of the links in ${\mathcal{E}_w^+}$ fails. Consequently, $(\mathcal{G}, \tilde{C}, \mathcal{R}, \lambda)$ is non-transferring.

\end{proof}

Theorem \ref{capincbad} shows that, for a broad family of locally routed network flows, increasing the link capacities does not necessarily result in a more resilient system.  An example
 where increasing the capacity of a link leads to a non-transferring system is illustrated in Fig. \ref{nrob}. Our next result complements Theorem \ref{capincbad} by characterizing the network topologies, for which increasing the link capacities would never cause a systemic failure under any local routing policy.

\begin{figure*}
\begin{center}
\includegraphics[trim =0mm 0mm 0mm 0mm,clip,scale=0.41]{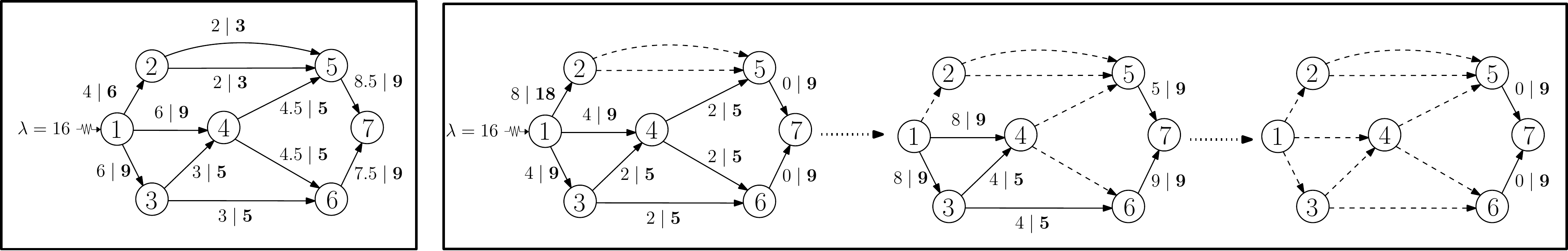}
\caption{A locally routed network flow under the proportional routing is shown. For each operational (solid) edge $e$, $f_e(t) \mid C_e(t)$ is provided next to $e$, and dashed edges denote that the corresponding edge has failed ($f_e(t) = C_e(t) =0$). The limiting state of a safely transferring system is shown on the left. The right side illustrates  some instants from the trajectory of the non-transferring system obtained by only increasing the capacity of the edge $(1,2)$.
}
\label{nrob}
\end{center}
\end{figure*}

\begin{theorem}
\label{capincbad2}
Let $\mathcal{G}= (\mathcal{V},\mathcal{E})$ be an acyclic flow network such that it contains no $v\in \mathcal{V}$ \textcolor{black}{with outgoing links to multiple nodes, at least one of which is an intermediate node (i.e., $|\mathcal{V}_v^+| \geq 2$, $ \mathcal{V}_v^+ \not \subseteq \mathcal{V}_D$)}. For any local routing policy $\mathcal{R}$, $\lambda \in \mathbb{R}_{+}^{\mathcal{V}_O}$, and $C \in \mathbb{R}_{+}^\mathcal{E}$, if  $(\mathcal{G}, C, \mathcal{R}, \lambda)$ is transferring, then $(\mathcal{G}, \tilde{C}, \mathcal{R}, \lambda)$ is also transferring for every $\tilde{C}  \geq C$.
\end{theorem}
\begin{proof}

Let $\mathcal{G}= (\mathcal{V},\mathcal{E})$ be an acyclic flow network such that there exists no $v\in \mathcal{V}$ for which $|\mathcal{V}_v^+| \geq 2$ and $ \mathcal{V}_v^+ \not \subseteq \mathcal{V}_D$. Then, for any $v\in \mathcal{V} \setminus \mathcal{V}_D$, either  ${\mathcal{V}_v^+ \subseteq \mathcal{V}_D}$, or a unique sequence of intermediate nodes should be traversed to reach any destination node from $v$. Accordingly, for any $\lambda \in \mathbb{R}_{+}^{\mathcal{V}_O}$ and $C \in \mathbb{R}_{+}^\mathcal{E}$,
(\ref{fbu}) implies that  $\mathcal{F}_\mathcal{G} (C,\lambda ) \neq \emptyset$ if and only if
\begin{equation}
\label{propeq1a}
\lambda_v  \leq \bold{1}^\mathsf{T}  C_{\mathcal{E}_v^+}, \forall v \in \mathcal{V}_O,
\end{equation}
\begin{equation}
\label{propeq1b}
 \sum_{u \in \mathcal{V}_O: v \in \mathcal{V}_R(u)}\lambda_u \leq \bold{1}^\mathsf{T}  C_{\mathcal{E}_v^+}, \forall v \in \mathcal{V}_I,
\end{equation}
where $\mathcal{V}_R(u)$ denotes the set of nodes that can be reached from $u$. Since $|\mathcal{V}_v^+|=1$ for every $v\in \mathcal{V} \setminus \mathcal{V}_D$ such that ${\mathcal{V}_v^+ \not \subseteq \mathcal{V}_D}$, any $(\mathcal{G}, C, \mathcal{R}, \lambda)$ satisfying (\ref{propeq1a}) and (\ref{propeq1b}) has
\begin{equation}
\label{propeqmu}
\mu_v(t)= \lambda_v + \sum_{u \in \mathcal{V}_v^-} \mu_u(t-1), \forall v \in \mathcal{V} \setminus \mathcal{V}_D.
\end{equation}
Accordingly, there are two possibilities for any such locally routed network flow $(\mathcal{G}, C, \mathcal{R}, \lambda)$: 1) the system is non-transferring, or  2) the system is safely transferring. More specifically, $(\mathcal{G}, C, \mathcal{R}, \lambda)$ is (safely) transferring if and only if (\ref{propeq1a}) and (\ref{propeq1b}) hold. Since
\begin{equation}
\label{propeq2}
\bold{1}^\mathsf{T}  \tilde{C}_{\mathcal{E}_v^+} \geq \bold{1}^\mathsf{T}  C_{\mathcal{E}_v^+}, \; \forall v \in \mathcal{V}, \forall \tilde{C} \geq C  \in \mathbb{R}_{+}^{\mathcal{E}},
\end{equation}
we obtain that, for any such network $\mathcal{G}$, ${\tilde{C} \geq C  \in \mathbb{R}_{+}^{\mathcal{E}}}$, ${\lambda \in \mathbb{R}_+^{\mathcal{V}_O} }$, and local routing policy $\mathcal{R}$, if $(\mathcal{G}, C, \mathcal{R}, \lambda)$ is transferring, then $(\mathcal{G}, \tilde{C}, \mathcal{R}, \lambda)$ is also transferring.
\end{proof}

\begin{remark}: Theorems \ref{capincbad} and \ref{capincbad2} together provide an exact characterization of network topologies for which a systemic failure can be induced under local routing by increasing the link capacities. For instance, in networks with a single destination node, having no $v\in \mathcal{V}$ such that $|\mathcal{V}_v^+| \geq 2$ and $ \mathcal{V}_v^+ \not \subseteq \mathcal{V}_D$ is equivalent to having $|\mathcal{V}_v^+| = 1$ for all ${v\in \mathcal{V} \setminus \mathcal{V}_D}$. Hence, for acyclic single-destination networks, it follows from Theorems \ref{capincbad} and \ref{capincbad2} that increasing the link capacities can not make the system non-transferring under any local routing policy if and only if the network is a directed tree (possibly with parallel edges) rooted at the destination node as in Fig. \ref{dipath}. For all other single-destination acyclic networks, increasing the link capacities may cause a systemic failure under some local routing policies (e.g., capacity-monotone policies) as shown in Theorem \ref{capincbad}.
\end{remark}

\begin{figure}[htb]
\begin{center}
\includegraphics[trim =0mm 0mm 0mm 0mm,clip,scale=0.58]{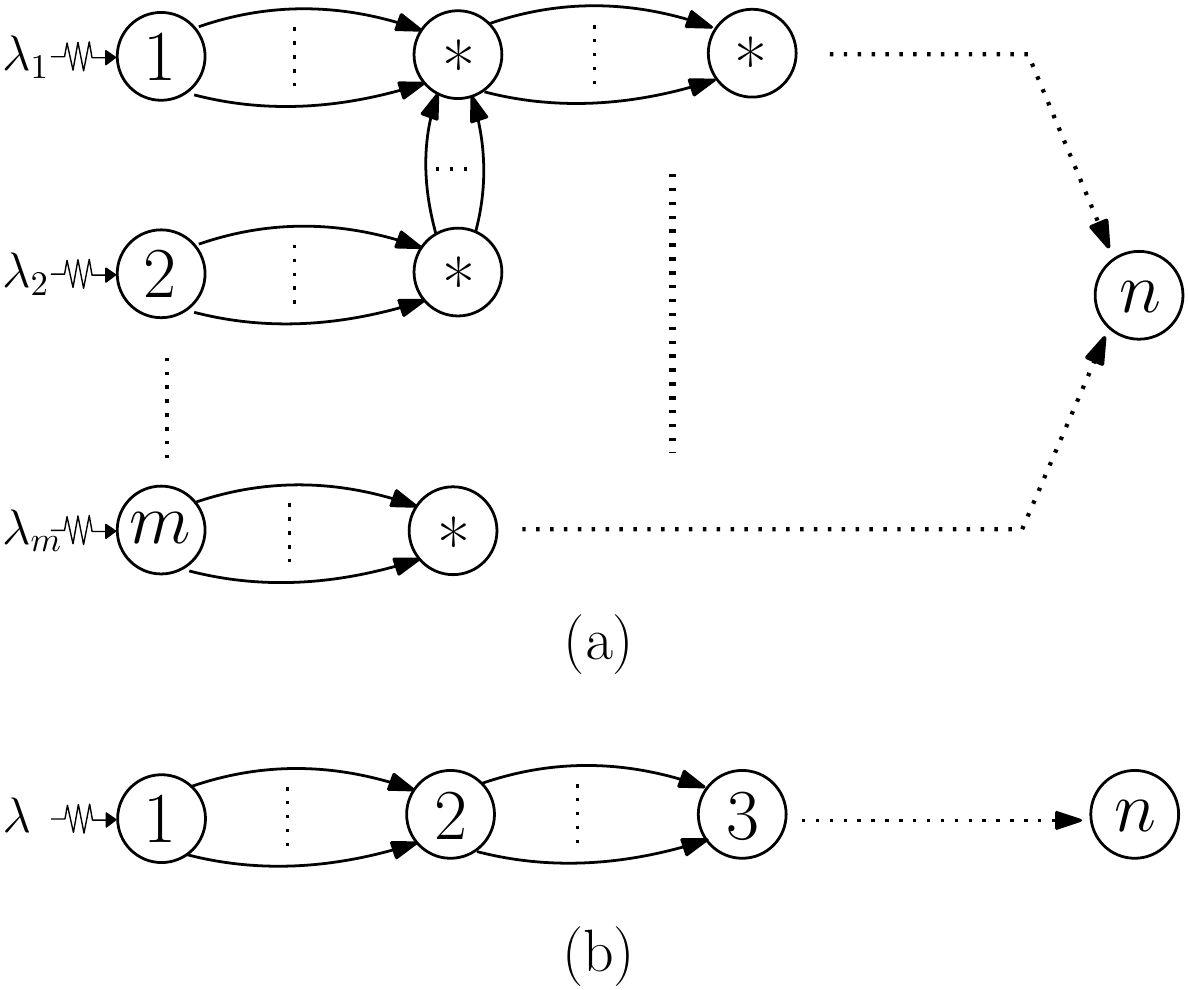}
\caption{Flow networks with a directed tree (possibly with parallel edges) topology are illustrated for the multi-origin and single-origin cases in (a) and (b) respectively. For each origin node of such a network, a unique sequence of nodes has to be traversed to reach the destination node. For acyclic networks with a single destination node, increasing the link capacities can not make the system non-transferring under any local routing policy if and only if the network has such a topology.}
\label{dipath}
\end{center}
\end{figure}

 \subsection{Avoiding Cascading Failures via Capacity Reduction}


In general, any non-transferring locally routed network flow $(\mathcal{G}, C, \mathcal{R}, \lambda)$ belongs to one of the two categories based on the existence of feasible balanced flows: 1) $\mathcal{F}_\mathcal{G} (C,\lambda ) = \emptyset$, and 2) $\mathcal{F}_\mathcal{G} (C,\lambda) \neq \emptyset$.  Note that if $\mathcal{F}_\mathcal{G} (C,\lambda ) = \emptyset$, then it is not possible to transfer the external inflow $\lambda$ without increasing the link capacities $C$ since (\ref{fb}) implies
\begin{equation}
\label{capmon}
\mathcal{F}_\mathcal{G} (\tilde{C},\lambda)\subseteq \mathcal{F}_\mathcal{G} (C,\lambda), \;  \forall C \geq \tilde{C} \in \mathbb{R}_{+}^\mathcal{E}, \forall \lambda \in \mathbb{R}_{+}^{\mathcal{V}_O}.
\end{equation} 

On the other hand,  if $\mathcal{F}_\mathcal{G} (C,\lambda) \neq \emptyset$, then there exist some feasible balanced flows, however none of them emerges under the local routing policy $\mathcal{R}$. In this case, the systemic failure originates from the local routing decisions rather than the discrepancy between the external inflow and the network capacity. 

Systemic failures due to local routing decisions may potentially be avoided by making the global network state accessible. For instance, GPS-based route guidance systems are used in transportation to provide the drivers with real-time traffic data. 
However, in distributed systems, the availability of more information does not guarantee the desired global behavior in the absence of a properly coordinated response (e.g., \cite{Arnott91,Acemoglu16,Roozbehani12}). In Theorem \ref{redt}, we will show that the failures due to local routing decisions can actually be avoided by appropriately reducing the link capacities. For instance, such capacity reductions can be realized by controlling the speed limits and traffic lights in transportation networks. 

For any acyclic flow network $\mathcal{G}= (\mathcal{V},\mathcal{E})$ and $C \in \mathbb{R}_{+}^\mathcal{E}$, let $\Lambda(\mathcal{G},C)$ denote set of external inflows for which there exists a feasible balanced flow, i.e., 
\begin{equation}
\label{tlambda}
\Lambda(\mathcal{G},C)=\{\lambda \in \mathbb{R}_{+}^{\mathcal{V}_O} \mid \mathcal{F}_\mathcal{G} (C,\lambda) \neq \emptyset \}.
\end{equation}
\textcolor{black}{ An acyclic flow network $\mathcal{G}$ with the link capacities $C$ is free of any link failures induced by local routing decisions if $(\mathcal{G}, C, \mathcal{R}, \lambda)$ is safely transferring for any local routing policy $\mathcal{R}$ and $\lambda \in \Lambda(\mathcal{G},C)$. Our next result shows that for any feasible external inflow $\lambda \in\Lambda(\mathcal{G},C)$, there exists $\undertilde{C} \in [\bold{0},C]$ such that the resulting network is free of any link failures due to local routing decisions and $[\bold{0},\lambda] \subseteq \Lambda(\mathcal{G},\undertilde{C})$. } 


\begin{theorem}
\label{redt}
For any acyclic flow network  $\mathcal{G}= (\mathcal{V},\mathcal{E})$, link capacities $C \in \mathbb{R}_{+}^\mathcal{E}$, and external inflow ${\lambda \in \Lambda(\mathcal{G},C)}$, there exists \textcolor{black}{$\undertilde{C} \in [\bold{0},C]$} such that $(\mathcal{G}, \undertilde{C}, \mathcal{R}, \lambda')$ is safely transferring for any local routing policy $\mathcal{R}$ and \textcolor{black}{${\lambda' \in \Lambda(\mathcal{G},\undertilde{C}) \supseteq [\bold{0},\lambda]}$}. 

\end{theorem}
\begin{proof}
For any acyclic flow network  $\mathcal{G}= (\mathcal{V},\mathcal{E})$, link capacities $C \in \mathbb{R}_{+}^\mathcal{E}$, and external inflow $\lambda \in \Lambda(\mathcal{G},C)$, let $\undertilde{C}=f$ for some feasible balanced flow vector $f \in \mathcal{F}_\mathcal{G} (C, \lambda)$. Then, due to (\ref{fb}),
\begin{equation}
\label{redt1}
\bold{0} \leq \undertilde{C} \leq C,
\end{equation}
\begin{equation}
\label{redt2}
\bold{1}^\mathsf{T} \undertilde{C}_{\mathcal{E}_v^-} = \bold{1}^\mathsf{T} \undertilde{C}_{\mathcal{E}_v^+}, \forall v \in  \mathcal{V}_I,
\end{equation} 
\begin{equation}
\label{redt3}
\bold{1}^\mathsf{T} \undertilde{C}_{\mathcal{E}_v^+} = \lambda_v, \; \forall v\in \mathcal{V}_O.
\end{equation}
In light of Lemma \ref{transt}, (\ref{redt2}) and (\ref{redt3}) together imply that $(\mathcal{G}, \undertilde{C}, \mathcal{R}, \lambda')$ is safely transferring for any local routing policy $\mathcal{R}$ and $\lambda' \in [\bold{0},\lambda]$.  \textcolor{black}{Furthermore, due to \eqref{fb}, \eqref{redt2}, and \eqref{redt3},
\begin{equation}
\label{redt4}
\mathcal{F}_\mathcal{G} (\undertilde{C},\lambda') \neq \emptyset  \Leftrightarrow  \lambda'\in [\bold{0},\lambda], \; \forall \lambda'\in \mathbb{R}_{+}^{\mathcal{V}_O}.
 \end{equation}
Consequently, $ \Lambda(\mathcal{G},\undertilde{C}) = [\bold{0},\lambda]$, and $(\mathcal{G}, \undertilde{C}, \mathcal{R}, \lambda')$ is safely transferring for any local routing policy $\mathcal{R}$ and $\lambda'\in \Lambda(\mathcal{G},\undertilde{C})$.}

\end{proof}


As a potential application of this result, suppose that for any given acyclic flow network ${\mathcal{G}= (\mathcal{V},\mathcal{E})}$, available link capacities $\bar{C} \in \mathbb{R}_{+}^\mathcal{E}$, and maximal inflow $\bar{\lambda} \in  \Lambda(\mathcal{G},\bar{C})$, it is desired to find some optimal capacity allocation $C \in [\bold{0},\bar{C}]$ that ensures the safe transfer of any external inflow $\lambda \in [\bold{0},\bar{\lambda}]$ under any local routing policy $\mathcal{R}$. In light of Lemma \ref{transt}, such a problem can be formulated as 
 \begin{equation}
 \label{opt}
 \underset{C \in \mathcal{C}_\mathcal{G}(\bar{C},\bar{\lambda})}{\text{min }}   J(C), \end{equation}
where $J(C) : \mathbb{R}_{+} ^\mathcal{E} \mapsto \mathbb{R}$ denotes some cost function, and
\begin{multline}
\label{cset}
\mathcal{C}_\mathcal{G}(\bar{C},\bar{\lambda})=\{C \in \mathbb{R}_{+} ^\mathcal{E} \mid C \leq \bar{C},  \;  \bar{\lambda}_v \leq \bold{1}^\mathsf{T} C_{\mathcal{E}_v^+}, \; \forall v\in \mathcal{V}_O  \\ 
\bold{1}^\mathsf{T} C_{\mathcal{E}_v^-} \leq \bold{1}^\mathsf{T} C_{\mathcal{E}_v^+}, \forall v \in  \mathcal{V}_I\}
\end{multline}
denotes the set of capacities $C \in [\bold{0},\bar{C}]$ that satisfy (\ref{transt2}) and (\ref{transt1}) for any $\lambda \in [\bold{0},\bar{\lambda}]$. In the remainder of this section, we present some examples of such capacity allocation problems.

\emph{Example (Minimal Reduction Problem)}: In some applications, it may be desired to ensure safety through a minimal reduction in the link capacities.  Such an objective may capture the effort needed to realize the capacity reduction, or the influence of reduction on some performance measure. For instance, in transportation networks, higher amounts of capacity reduction may yield higher average delay. One way of formulating the minimal reduction problem is to set
 \begin{equation}
 \label{costj}
 J(C) = \alpha^\mathsf{T} (\bar{C}-C),
 \end{equation}
where $\alpha \in  \mathbb{R}_{+}^\mathcal{E}$ denotes the cost associated with reducing the link capacities. The problem in (\ref{opt}) is a linear program with the cost function in (\ref{costj}).

\emph{Example (Robust Safety Problem)}: In many systems, the network may be subjected to some exogenous disturbances on the link capacities. Accordingly, it may be desired to allocate the link capacities to maximize the amount of disturbance that can be absorbed without losing the safety guarantee. One way of achieving this objective is to maximize the minimum amount of capacity loss that can take the allocated capacities outside the feasible region of (\ref{opt}), i.e.,

\begin{eqnarray}
 \label{robj}
 J(C) & = & - \min_{ \delta \in [\bold{0},C]:   (C - \delta) \notin \mathcal{C}_\mathcal{G}(\bar{C},\bar{\lambda})} \bold{1}^\mathsf{T} \delta \nonumber \\
  & = &  \max_{v \in \mathcal{V}\setminus  \mathcal{V}_D } \bar{\lambda}_v + \bold{1}^\mathsf{T} C_{\mathcal{E}_v^-} - \bold{1}^\mathsf{T} C_{\mathcal{E}_v^+},
\end{eqnarray}
where $\mathcal{C}_\mathcal{G}(\bar{C},\bar{\lambda})$ denotes the feasible region as given in (\ref{cset}). Note that (\ref{opt}) is a convex problem with the cost in (\ref{robj}).

In Fig. \ref{capallocex}, we demonstrate the solutions to the minimal reduction and robust safety problems for a sample flow network. The network $\mathcal{G}=(V,E)$ is shown in Fig. \ref{capallocex}a along with the available link capacities $\bar{C}$ and the maximal external inflow $\bar{\lambda}$ to be transferred. In this example, allocating all the available capacity may lead to a non-transferring system under some external inflow $\lambda \in [\bold{0}, \bar{\lambda}]$ and local routing policy $\mathcal{R}$. For instance, Fig. \ref{capallocex}b shows the resulting equilibrium  under the proportional routing policy for $C=\bar{C}$ and $\lambda=\bar{\lambda}$. The minimal reduction problem with $\alpha= \bold{1}$ leads to an optimal $C \in [\bold{0},\bar{C}]$ with ${\bold{1}^\mathsf{T} (\bar{C}-C)=4}$ as shown in Fig. \ref{capallocex}c, which ensures that $(\mathcal{G}, C, \mathcal{R}, \lambda)$ is safely transferring for any local routing policy $\mathcal{R}$ and $\lambda \in [\bold{0}, \bar{\lambda}]$. The robust safety problem has an optimal $C \in [\bold{0},\bar{C}]$ as shown in Fig. \ref{capallocex}d, which ensures that ${(\mathcal{G}, C-\delta, \mathcal{R}, \lambda)}$ is safely transferring for any local routing policy $\mathcal{R}$, $\lambda \in [\bold{0}, \bar{\lambda}]$, and $\delta \in [\bold{0},C]$ such that ${\bold{1}^\mathsf{T} \delta \leq 0.4}$.

\begin{figure}[htb]
\begin{center}
\includegraphics[trim =0mm 0mm 0mm 0mm,clip,scale=0.86]{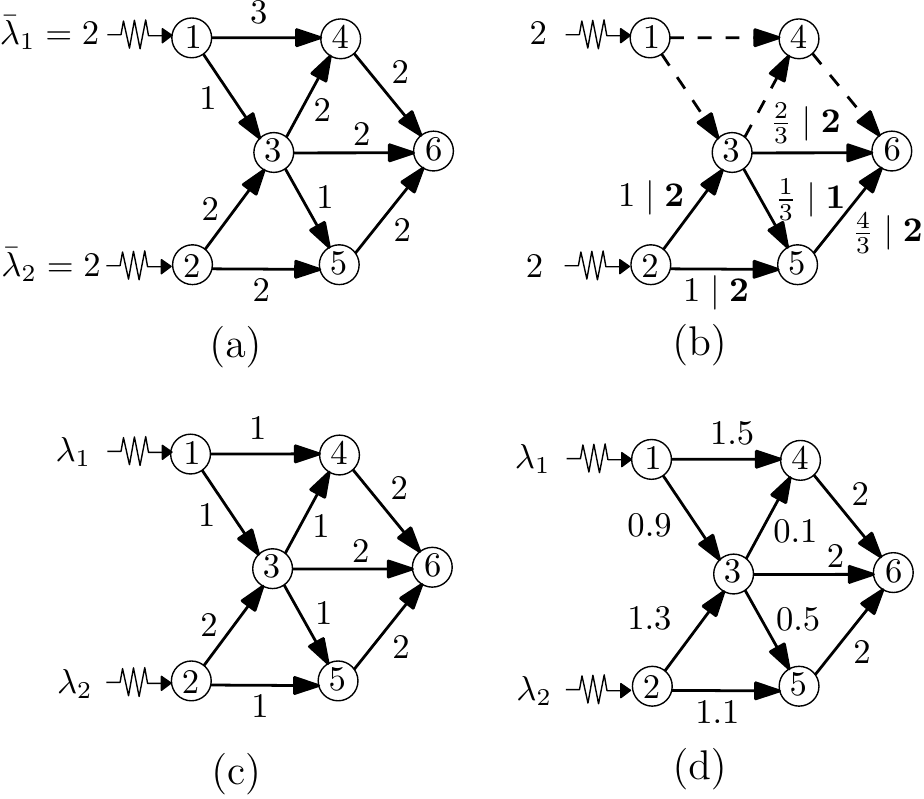}
\caption{An acyclic flow network $\mathcal{G}= (\mathcal{V},\mathcal{E})$ is shown in (a), where the available link capacities $\bar{C}$ are provided next to the edges, and the maximal external inflow to be transferred is $\bar{\lambda}= [2 \; 2]^\mathsf{T}$. The capacities in (a) may lead to a non-transferring system under some external inflow $\lambda \in [\bold{0}, \bar{\lambda}]$ and local routing policy $\mathcal{R}$. For instance, the limiting state $(C^*,f^*)$ under the proportional routing policy for $\lambda=[2 \; 2]^\mathsf{T}$ is shown in (b), where $f^*_e \mid C^*_e$ is provided next to each operational (solid) edge $e$.  A solution to the minimal reduction problem (see (\ref{opt}), (\ref{cset}), (\ref{costj})) for $\alpha=\bold{1}$ is shown in (c), and a solution to the robust safety problem (see (\ref{opt}), (\ref{cset}), (\ref{robj})) is shown in (d). }

\label{capallocex}
\end{center}
\end{figure}

%
%
%
%

\section{Conclusions and Future Work}
\label{conclusion}
In this paper, we showed that the link capacities do not have a monotonic influence on the resilience of locally routed network flows with finite link capacities. The local routing policies were defined as mappings from the local inflows and link capacities to the corresponding equilibrium outflows. Accordingly, a link is overloaded only if its tail receives a total inflow larger than the total capacity of its outgoing links. A link irreversibly fails if it is overloaded, or if there is no operational link in its immediate downstream to carry its flow. We showed that, under such dynamics, increasing the link capacities can cause a systemic failure for a broad family of network topologies and local routing policies. Furthermore, when the external inflow does not exceed the network capacity, the failures that arise from local routing decisions can actually be avoided by appropriately reducing the link capacities.

As a future direction, we plan to build on the results of this paper for designing dynamic capacity allocation policies to avert the systemic failures due to local routing decisions in network flows. Furthermore, we plan to investigate how the availability of more information (e.g., state of the partial downstream) in the routing policies would affect the influence of link capacities on resilience.




\bibliography{MyReferences}
\end{document}